\documentclass[a4paper,11pt]{article}
\usepackage{amssymb,amsmath,latexsym,amsthm}
\usepackage{color,mathrsfs}
\usepackage{url}
\usepackage{enumerate}
\usepackage[square,sort,comma,numbers]{natbib}
\usepackage{dsfont}
\usepackage{diagbox}
 \usepackage{geometry}
\usepackage{graphicx}
\usepackage{colortbl}
\usepackage[nottoc,numbib]{tocbibind}
\usepackage{endnotes}

\renewcommand{\textbf}[1]{\begingroup\bfseries\mathversion{bold}#1\endgroup}

\usepackage{fancyhdr}
\usepackage{pstricks,pst-plot,pst-node,pstricks-add}

\newtheorem{thm}{Theorem}[section]
\newtheorem{defi}{Definition}[section]
\newtheorem{corollary}[thm]{Corollary}
\newtheorem{prop}[thm]{Proposition}

\theoremstyle{definition}
\newtheorem{remark}[thm]{Remark}

\newtheorem{example}[thm]{Example}
\newtheorem{examples}[thm]{Examples}

\newlength{\bibitemsep}\newlength{\bibparskip}
\let\oldthebibliography\thebibliography \renewcommand\thebibliography[1]{
  \oldthebibliography{#1} \setlength{\parskip}{\bibitemsep}
  \setlength{\itemsep}{\bibparskip} }

\newcommand{\R}{\mathbb R}

\newcommand{\Z}{\mathbb Z}
\newcommand{\N}{\mathbb N}

\numberwithin{equation}{section}

\def\XXint#1#2#3{{\setbox0=\hbox{$#1{#2#3}{\int}$}
    \vcenter{\hbox{$#2#3$}}\kern-.5\wd0}}

\allowdisplaybreaks
\date{date}
\begin{document}
\title{Effect of periodic arrays of defects on lattice energy minimizers}
\author{Laurent B\'{e}termin\\ \\
Faculty of Mathematics, University of Vienna,\\ Oskar-Morgenstern-Platz 1, 1090 Vienna, Austria\\ \texttt{laurent.betermin@univie.ac.at}. ORCID id: 0000-0003-4070-3344 }
\date\today
\maketitle

\begin{abstract}
We consider interaction energies $E_f[L]$ between a point $O\in \mathbb{R}^d$, $d\geq 2$, and a lattice $L$ containing $O$, where the interaction potential $f$ is assumed to be radially symmetric and decaying sufficiently fast at infinity. We investigate the conservation of optimality results for $E_f$ when integer sublattices $k L$ are removed (periodic arrays of vacancies) or substituted (periodic arrays of substitutional defects). We consider separately the non-shifted ($O\in k L$) and shifted ($O\not\in k L$) cases and we derive several general conditions ensuring the (non-)optimality of a universal optimizer among lattices for the new energy including defects. Furthermore, in the case of inverse power laws and Lennard-Jones type potentials, we give necessary and sufficient conditions on non-shifted periodic vacancies or substitutional defects for the conservation of minimality results at fixed density. Different examples of applications are presented, including optimality results for the Kagome lattice and energy comparisons of certain ionic-like structures. 
\end{abstract}


\tableofcontents

\section{Introduction, setting and goal of the paper}

\subsection{Lattice energy minimization and setting}

Mathematical results for identifying the lattice ground states of interacting systems have recently attracted a lot of attention. Even though the `Crystallization Conjecture' \cite{BlancLewin-2015} -- the proof of existence and uniqueness of a periodic minimizer for systems with free particles -- is still open in full generality, many interesting results have been derived in various settings for showing the global minimality of certain periodic structures including the uniform chain $\Z$, the triangular lattice $\mathsf{A}_2$, the square lattice $\Z^2$, the face-centred cubic lattice $\mathsf{D}_3$ (see Fig. \ref{fig:Lattices}), as well as the other best packings $\mathsf{E}_8$ and the Leech lattice $\Lambda_{24}$ (see \cite{BDLPSquare,CKMRV2Theta} and references therein). Moreover, the same kind of investigation has been made for multi-component systems (e.g. in \cite{Friedrich:2018aa,FrieKreutSquare,LuoChenWei,BFK20,LuoWeiBEC}) where the presence of charged particles yield to rich energetically optimal structures. These problems of optimal point configurations are known to be at the interface of Mathematical Physics, Chemistry, Cryptography, Geometry, Signal processing, Approximation, Arithmetic, etc. The point of view adopted in this work is the one of Material Science where the points are thought as particles or atoms.

\begin{figure}[!h]
\centering
\includegraphics[width=7cm]{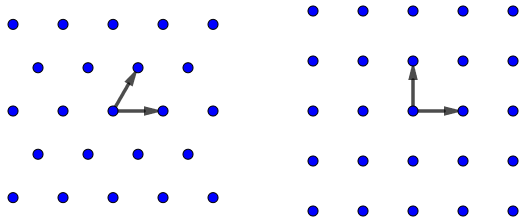} \qquad\includegraphics[width=3cm]{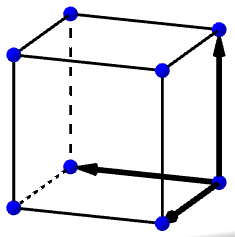}\qquad\includegraphics[width=3cm]{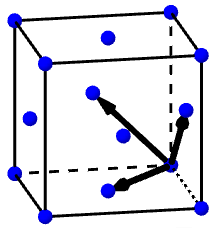}
\caption{\small{In dimension $d=2$, representation of the triangular and square lattices respectively  defined by $\mathsf{A}_2=\lambda_1\left[\Z(1,0)\oplus \Z(1/2,\sqrt{3}/2) \right]$ and $\Z^2$. In dimension $d=3$, representation of the simple cubic and the face-centred cubic lattices respectively defined by $\Z^3$ and $\mathsf{D}_3:=\lambda_2\left[\Z(1,0,1)\oplus \Z(0,1,1)\oplus \Z(1,1,0)  \right]$. The constants $\lambda_1,\lambda_2$ are such that the lattices have unit density.}}
\label{fig:Lattices}
\end{figure}

In this paper, our general goal is to show mathematically how the presence of periodic arrays of charges (called here `defects' in contrast with the initial crystal `atoms') in a perfect crystal affects the minimizers of interaction energies when the interaction between species is radially symmetric. Since the structure of crystals are often given by the same kind of lattices, it is an important question to know the conditions on the added periodic distribution of defects and on the interaction energy in order to have conservation of the initial ground state structure. Furthermore, only very few rigorous results are available on minimization of charged structures among lattices (see e.g. our recent works \cite{BFK20,MaxTheta1}).

\medskip

We therefore assume the periodicity of our systems, and once we restrict this kind of problem to the class of (simple) lattices and radially symmetric interaction potentials, an interesting non-trivial problem is to find the minimizers of a given energy per point among these simple periodic sets of points, with or without a fixed density. In this paper, we keep the same kind of notations we have used in our previous works (see e.g. \cite{OptinonCM,BetSoftTheta,BFK20}). More precisely, for any $d\geq 2$ we called $\mathcal{L}_d$ the class of $d$-dimensional lattices, i.e. discrete co-compact subgroups or $\R^d$, 
$$
\mathcal{L}_d:=\left\{ L=\bigoplus_{i=1}^d \Z u_i :\textnormal{$\{u_1,...,u_d\}$ is a basis of $\R^d$}\right\},
$$
and, for any $V>0$, $\mathcal{L}_d(V)\subset \mathcal{L}_d$ denotes the set of lattices with volume $|\det(u_1,...,u_d)|=V$, i.e. such that its unit cell $Q_L$ defined by
\begin{equation}\label{eq:QL}
Q_L:=\left\{x=\sum_{i=1}^d \lambda_i u_i : \forall i\in \{1,...,d\}, \lambda_i\in [0,1)  \right\},
\end{equation}
has volume $|Q_L|=V$. We will also say that $L\in \mathcal{L}_d(V)$ has density $V^{-1}$. The class $\mathcal{F}_d$ of radially symmetric functions we consider is, calling $\mathcal{M}_d$ the space of Radon measures on $\R_+$,
$$
\mathcal{F}_d:= \left\{f:\R_+\to \R : f(r)=\int_0^{\infty} e^{-rt}d\mu_f(t), \mu_f \in \mathcal{M}_d,|f(r)|=O(r^{-p_f}) \textnormal{ as $r\to \infty$, $p_f>d/2$}\right\}.
$$
When $\mu_f$ is non-negative, $f$ is a completely monotone function, which is equivalent by Hausdorff-Bernstein-Widder Theorem \cite{Bernstein-1929} with the property that for all $r>0$ and all $k\in \N$, $(-1)^k f^{(k)}(r)\geq 0$. We will write this class of completely monotone functions as
$$
\mathcal{F}^{cm}_d:=\left\{ f\in \mathcal{F}_d : \mu_f\geq 0  \right\}.
$$

For any $f\in \mathcal{F}_d$, we thus defined the $f$-energy $E_f[L]$ of a lattice $L$, which is actually the interaction energy between the origin $O$ of $\R^d$ and all the other points of $L$, by
\begin{equation}\label{eq:defEF}
E_f[L]:=\sum_{p\in L\backslash \{0\}} f(|p|^2).
\end{equation}
Notice that this sum is absolutely convergent as a simple consequence of the definition of $\mathcal{F}_d$. We could also define $E_f$ without such decay assumption by renormalizing the sum using, for instance, a uniform background of opposite charges (see e.g. \cite{LewinFloating}) or an analytic continuation in case of parametrized potential such as $r^{-s}$ (see \cite{Latticesums}). 

\medskip

One can interpret the problem of minimizing $E_f$ in $\mathcal{L}_d$ (or in $\mathcal{L}_d(V)$ for fixed $V>0$) as a geometry optimization problem for solid crystals where the potential energy landscape of a system with an infinite number of particles is studied in the restricted class of lattice structures. Even though the interactions in a solid crystal are very complex (quantum effects, angle-dependent energies, etc.), it is known that the Born-Oppenheimer adiabatic approximation used to describe the interaction between atoms or ions in a solid by a sum of pairwise contributions (see e.g. \cite[p. 33 and p. 945]{CondensMatter} and \cite{Tosi}) is a good model for `classical crystals' (compared to `quantum crystals' \cite{BuchananQC}), i.e. where the atoms are sufficiently heavy. Moreover, since all the optimality properties we are deriving in this paper are invariant under rotations, all the results will be tacitly considered up to rotations.

\medskip

 Furthermore, studying such interacting systems with this periodicity constraint is a good method to keep or exclude possible candidates for a crystallization problem (i.e. with free particles). We are in particular interested in a type of lattice $L_d$ that is the unique minimizer of $E_f$ in $\mathcal{L}_d(V)$ for any fixed $V>0$ and any completely monotone potential $f\in \mathcal{F}_d^{cm}$. Following Cohn and Kumar \cite{CohnKumar} notion (originally defined among all periodic configurations), we call this property the \textit{universal optimality among lattices} of $L_d$ (or \textit{universal optimality in $\mathcal{L}_d(1)$}). 

\medskip

Only few methods are available to carry out the minimization of $E_f$. Historically, the first one consists to parametrize all the lattices of $\mathcal{L}_d(1)$ in an Euclidean fundamental domain $\mathcal{D}_d\subset \R^{\frac{d(d+1)}{2}-1}$ (see e.g. \cite[Sec. 1.4]{Terras_1988}) and to study the variations of the energy in $\mathcal{D}_d$. It has been done in dimension 2 for showing the optimality of the triangular lattice $\mathsf{A}_2$ at fixed density for the Epstein zeta function \cite{Rankin,Eno2,Cassels,Diananda} and the lattice theta function \cite{Mont} respectively defined for $s>d$ and $\alpha>0$ by
\begin{equation}\label{eq:zetatheta}
\zeta_L(s):=\sum_{p\in L\backslash \{0\}} \frac{1}{|p|^s}, \qquad \textnormal{and}\qquad \theta_L(\alpha):=\sum_{p\in L} e^{-\pi \alpha |p|^2}.
\end{equation}
In particular, a simple consequence of Montgomery's result \cite{Mont} for the lattice theta function is the universal optimality among lattices of $\mathsf{A}_2$  (see e.g. \cite[Prop. 3.1]{BetTheta15}). Other consequences of the universal optimality of $\mathsf{A}_2$ among lattices have been derived for other potentials (including the Lennard-Jones one) \cite{Betermin:2014fy,BetTheta15,OptinonCM,LBMorse} as well as masses interactions \cite{BetKnupfdiffuse}. Furthermore, new interesting and general consequences of universal optimality will be derived in this paper, including a sufficient condition for the minimality of a universal minimizer at fixed density (see Theorem \ref{thm1}).

\medskip

This variational method is also the one we have recently chosen in \cite{MaxTheta1} for showing the maximality of $\mathsf{A}_2$ in $\mathcal{L}_2(1)$ -- and conjectured the same results in dimensions $d\in \{8,24\}$ for the lattices $\mathsf{E}_8$ and $\Lambda_{24}$ -- for the alternating and centered lattice theta function respectively defined, for all $\alpha>0$, by
\begin{equation}\label{eq:defthetaaltcent}
\theta_L^\pm(\alpha):=\sum_{p\in L} \varphi_\pm(p) e^{-\pi \alpha |p|^2},\quad \textnormal{and}\quad \theta_L^c(\alpha):=\sum_{p\in L} e^{-\pi\alpha |p+c_L|^2},
\end{equation}
where $L=\bigoplus_{i=1}^d\Z u_i$, $\{u_i\}_i$ being a Minkowski (reduced) basis of $L$ (see e.g. \cite[Sect. 1.4.2]{Terras_1988}), $\varphi_\pm(p):=\sum_{i=1}^d m_i$ for $p=\sum_{i=1}^d m_i u_i$, $m_i\in \Z$ for all $i$, and $c_L=\frac{1}{2}\sum_i u_i$ is the center of its unit cell $Q_L$. In particular, the alternate lattice theta function $\theta_L^\pm(\alpha)$ can be viewed as the Gaussian interaction energy of a lattice $L$ with an alternating distribution of charges $\pm 1$, which can be itself seen as the energy once we have removed $2$ times the union of sublattices $\cup_{i=1}^d (L+u_i)$ from the original lattice $L$. This result shows another example of universal optimality -- we will call it \textit{universal maximality} -- among lattices, i.e. the maximality of $\mathsf{A}_2$ in $\mathcal{L}_2(1)$ for the energies $E_f^\pm$ and $E_f^c$ defined by
\begin{equation}\label{eq:defaltercenter}
E_f^\pm[L]:=\sum_{p\in L\backslash \{0\}} \varphi_\pm(p) f(|p|^2),\quad \textnormal{or}\quad E_f^c[L]:=\sum_{p\in L} f(|p+c_L|^2),
\end{equation}
where $f\in \mathcal{F}_d^{cm}$. This kind of problem was actually our first motivation for investigating the effects of periodic arrays of defects on lattice energy minimizers, since removing two times the sublattices $2L+u_1$ and $2L+u_2$ totally inverse the type of optimality among lattices. Furthermore, this maximality result will also be used in Theorem \ref{thm02}, applied -- in the general case of a universal maximizer $L_d^\pm$ for $E_f^\pm$ in any dimension where this property could be shown -- for other potentials $\mathcal{F}_d\backslash \mathcal{F}_d^{cm}$ in Theorem \ref{thm02max} and compared with other optimality results in Section \ref{subsec:ions}.

\medskip

The second method for showing such optimality result is based on the Cohn-Elkies linear programming bound that was successfully used for showing the best packing results in dimensions 8 and 24 for $\mathsf{E}_8$ and $\Lambda_{24}$ in \cite{Viazovska,CKMRV}, as well as their universal optimality among all periodic configurations in \cite{CKMRV2Theta}. As in the two-dimensional case, many consequences of these optimality results have been shown for other potentials \cite{OptinonCM,PetracheSerfatyCryst} and masses interactions \cite{BetSoftTheta}.

\subsection{Problem studied in this paper and connection to Material Science}

The goal of this work is to investigate the effect on the minimizers of $E_f$ when we change, given a lattice $L\subset \mathcal{L}_d$ and $K\subset \N\backslash \{1\}$, a certain real number $a_k\neq 0$ of integer sublattices $k L$, $k \in K$, in the original lattice, and where the lattices $k L$ might be shifted by a finite number of lattice points $L_k:=\{p_{i,k}\}_{i\in I_k}\subset L$ for some finite set $I_k$. Writing 
\begin{equation}\label{eq:kappa}
\kappa:=\{K, A_K, P_K\},\quad K\subset \N\backslash \{1\}, \quad A_K=\{a_k\}_{k\in K}\subset \R^*,\quad P_K=\bigcup_{k\in K} L_k,\quad L_k=\{p_{i,k}\}_{i\in I_k}\subset L,
\end{equation}
the new energy $E^{\kappa}_f$ we consider, defined for $f\in \mathcal{F}_d$ and $\kappa$ as in \eqref{eq:kappa} and such that the following sum on $K$ is absolutely convergent, is given by
\begin{equation}
E^{\kappa}_f[L]:=E_f[L]-\sum_{k \in K} \sum_{i\in I_k} a_k E_f[p_{i,k}+k L].
\end{equation}
In particular, in the non-shifted case, i.e. $P_K=\emptyset$, then
\begin{equation}\label{eq:fAK}
E^{\kappa}_f[L]=E_{f_{\kappa}}[L],\quad \textnormal{where}\quad f_{\kappa}(r):=f(r)-\sum_{k\in K} a_k f(k^2 r).
\end{equation}

Since we are interested in the effects of defects on lattice energy ground states, we therefore want to derive conditions on $\kappa$ and $f$ such that $E_f$ and $E_f^\kappa$ have the same minimizers in $\mathcal{L}_d$ or $\mathcal{L}_d(V)$ for fixed $V>0$. In particular, we also want to know if the universal minimality among lattices of a lattice $L_d$ is conserved while removing or substituting integer sublattices. This a natural step for investigating the robustness of the optimality results stated in the previous section of this paper when the interaction potential is completely monotone or, for instance, of Lennard-Jones type. Furthermore, it is also the opportunity to derive new applications and generalizations of the methods recently developed in \cite{BetTheta15,OptinonCM,MaxTheta1} for more `exotic' ionic-like structures.

\medskip

Replacing integer sublattices as described above can be interpreted and classified in two relevant cases in Material Science:
\begin{enumerate}
\item If $a_k=1$, then removing only once the sublattice $kL$ from $L$ creates a periodic array of vacancies (also called periodic Schottky defects \cite[Sect. 3.4.3]{Tilley});
\item If $a_k\neq 1$, then `removing' $a_k$ times the sublattice $kL$ from $L$ creates a periodic array of substitutional defects (also called impurities), where the original lattice points (initially with charges $+1$) are replaced by points with `charges' (or `weights') $1-a_k\neq 0$.
\end{enumerate}

In Figure \ref{fig:ionic}, we have constructed three examples of two-dimensional lattices with periodic arrays of defects which certainly do not exist in the real world. In contrast, Figure \ref{fig:introionic} shows two important examples of crystal structures arising in nature: the Kagome lattice and the rock-salt structure. These examples are discussed further in Section \ref{subsec:examples}.

\medskip

\begin{figure}[!h]
\centering
\includegraphics[width=4cm]{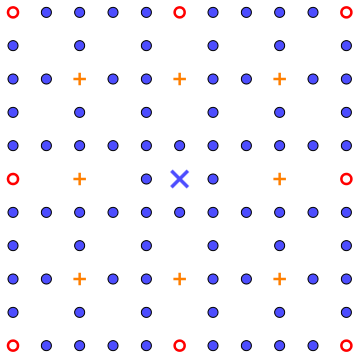} \qquad \includegraphics[width=5cm]{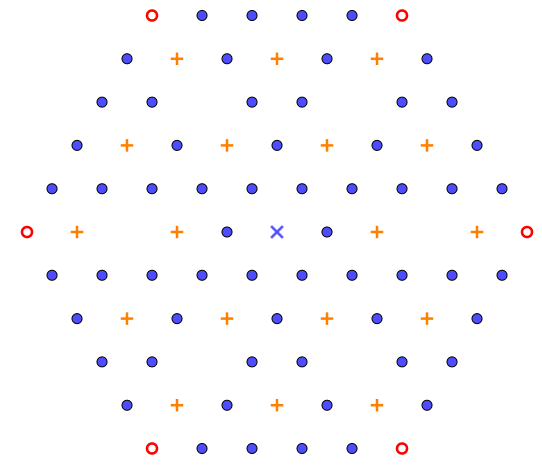} \qquad \includegraphics[width=42mm]{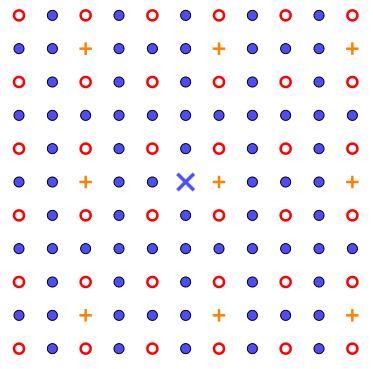} 
\caption{\small{Mathematical examples of periodic array of defects performed on a patch of the square lattice $\Z^2$ (left and right) and the triangular lattice $\mathsf{A}_2$ (middle). The cross $\textcolor{blue}{\times}$ represents the origin $O$ of $\R^2$. The points marked by $\textcolor{blue}{\bullet}$ are the original points of the lattice whereas the points marked by $\textcolor{orange}{+}$ and $\textcolor{red}{\circ}$ are substitutional defects of charge $1-a_k$ for some $a_k\in \R^*\backslash \{1\}$ and some $k\in K_:=\{2,3,4,5 \}$. The missing lattice points are the vacancy defects. The patch on the right contains two shifted periodic arrays of defects.}}
\label{fig:ionic}
\end{figure}

\begin{figure}[!h]
\centering
\includegraphics[width=5cm]{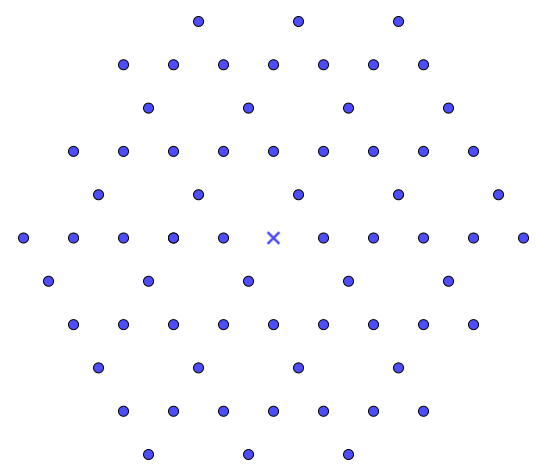} \qquad \includegraphics[width=4cm]{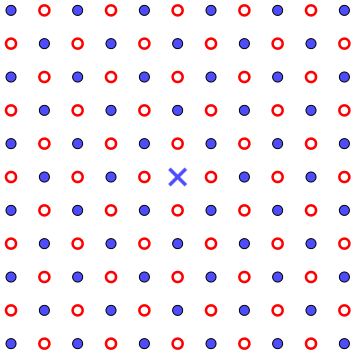}
\caption{\small{Two examples of 2d lattices patches with a periodic array of defect arising in nature. The left-hand structure is the Kagome lattice obtained by removing from the triangular lattice $\mathsf{A}_2$ the sublattice $2 \mathsf{A}_2 + (1,0)+(1/2,\sqrt{3}/2)$. It appears to be a layer of the jarosite. The right-hand structure is the 2d rock-salt structure obtained by removing from the square lattice $\Z^2$ two times the sublattices $2\Z^2 + (1,0)$ and $2\Z^2 + (0,1)$ in such a way that particles of opposites signs $\pm 1$ alternate ($\textcolor{blue}{\bullet}$ and $\textcolor{red}{\circ}$ correspond respectively to charges of signs $1$ and $-1$). It is itself a layer of the three-dimensional rock-salt structure NaCl.}}
\label{fig:introionic}
\end{figure}

 While the substitutional defects case has different interpretations and applications in terms of optimal multi-component (ionic) crystals (see e.g. Section \ref{subsec:ions}), the vacancy case is also of interest when we look for accelerating the computational time for checking numerically the minimality of a structure. Indeed, if the minimizer does not change once several periodic arrays of points are removed from all lattices, then a computer will be faster to check this minimality. This is of practical relevance in particular in low dimensions since the computational time of such lattice energies, which grows exponentially with the dimension, are extremely long in dimension $d\geq 8$ -- even with the presence of periodic arrays of vacancies -- and shows how important are rigorous minimality results in these cases.

\medskip

Furthermore, from a Physics point of view, it is well-known (see e.g. \cite{Tilley}) that point defects play an important role in crystal properties. As explained in \cite{Allaire}: `Crystals are like people, it is the defects in them which tend to make them interesting'. For instance, they reduce the electric and thermal conductivity in metals and modify the colors of solids and their mechanical strength. We also notice that substitutional defects control the electronic conductivity in semi-conductors, whereas the vacancies control the diffusion and the ionic conductivity in a solid. In particular, there is no perfect crystal in nature and it is then interesting and physically relevant to study optimality results for periodic systems with defects, in particular for models at positive temperature where the number of vacancies per unit volume increases exponentially with the temperature (see e.g. \cite[Sec. 3.4.3]{Tilley}). Notice that the raise of temperature also creates another kind of defects called self-interstitial -- i.e. the presence of extra atoms out of lattice sites -- but they are known to be negligible (at least if they are of the same type than the solid's atoms) compared to the vacancies when disorder appears, excepted for Silicon.

\medskip

\textbf{Plan of the paper.} Our main results are presented in Section \ref{subsec:main} whereas their proofs are postponed to Section \ref{sec:proof}. Many applications of our results are discussed in Section \ref{subsec:examples}, including explicit examples of minimality results for the Kagome lattice and other ionic structures.

\section{Statement of the main results}\label{subsec:main}

\subsection{On the minimality of a universal optimizer}

We start by recalling the notion of universal optimality among lattices as defined by Cohn and Kumar in \cite{CohnKumar}.

\begin{defi}[\textbf{Universal optimality among lattices}]
Let $d\geq 2$. We say that $L_d$ is universally optimal in $\mathcal{L}_d(1)$ if $L_d$ is a minimizer of $E_f$ defined by \eqref{eq:defEF} in $\mathcal{L}_d(1)$ for any $f\in \mathcal{F}_d^{cm}$.
\end{defi}

\begin{remark}[Universally optimal lattices]
We recall again that the only known universally optimal lattices in dimension $d\geq 2$ are $\mathsf{A}_2$ (see \cite{Mont}), $\mathsf{E}_8$ and the Leech lattice $\Lambda_{24}$ (see \cite{CKMRV2Theta}) in dimensions $d\in \{2,8,24\}$. It is also shown in \cite[p. 117]{SarStromb} that there is no such universally optimal lattice in dimension $d=3$. There are also clear indications (see \cite[Sect. 6.1]{OptinonCM}) that the space of functions for which the minimality at all the scales of $L_d$ holds is much larger than $\mathcal{F}_d^{cm}$.
\end{remark}

Before stating our results, notice that all of them are stated in terms of global optimality but could be rephrased for showing local optimality properties. This is important, in particular in dimensions $d=3$ where only local minimality results are available for $E_f$ (see e.g. \cite{Beterminlocal3d}) and can be generalized for energies of type $E_f^\kappa$, ensuring the local stability of certain crystal structures.

\medskip

We now show that the universal optimalities among lattices in dimension $d\in \{2,8,24\}$ proved in \cite{Mont,CKMRV2Theta} are not conserved in the non-shifted case once we only removed a single integer sublattice a positive number $a_k>0$ of times, whereas they are conserved when $a_k<0$.

\begin{thm}[\textbf{Conservation of universal optimalities - Non-shifted case}]\label{thm0}
Let $f$ be defined by $f(r)=e^{-\pi \alpha r}$, $\alpha>0$. For all $d\in \{2,8,24\}$, all $k\in \N\backslash \{1\}$, all $a_k>0$ and $\kappa=\{k, a_k, \emptyset\}$, there exists $\alpha_d$ such that for all $\alpha \in (0,\alpha_d)$, $\mathsf{A}_2$, $\mathsf{E}_8$ and the Leech lattice $\Lambda_{24}$ are not minimizers of $E_f^{\kappa}$ in $\mathcal{L}_d(1)$. \\
Furthermore, for any $d\in \{2,8,24\}$, for any $K\subset \N\backslash \{1\}$, any $A_K=\{a_k\}_{k\in K}\subset \R_-$ and $\kappa=\{K, A_K,\emptyset\}$, $\mathsf{A}_2$, $\mathsf{E}_8$ and the Leech lattice $\Lambda_{24}$ are the unique minimizers of $E_f^\kappa$ in $\mathcal{L}_d(1)$ for all $\alpha>0$.
\end{thm}
\begin{remark}[Generalization to $4$-designs]
The non-optimality result in Theorem \ref{thm0} is obtained by using the Taylor expansion of the theta function found by Coulangeon and Sch\"urmann in \cite[Eq. (21)]{Coulangeon:2010uq}. Therefore, the result is actually generalizable to any universal optimal lattice $L_d$ such that all its layers (or shells) are $4$-designs, i.e. such that for all $r>0$ with $\{\partial B_r\cap L_d  \}\neq \emptyset$, $B_r$ being the ball centred at the origin and with radius $r$, and all polynomial $P$ of degree up to $4$ we have
$$
\frac{1}{|\partial B_r|}\int_{\partial B_r} P(x)dx=\frac{1}{\sharp\{\partial B_r\cap L_d  \}}\sum_{x\in \partial B_r\cap L_d} P(x).
$$
\end{remark}

We now present a sufficient condition on $P_K$ such that the triangular lattice is universally optimal in $\mathcal{L}_2(1)$ for $E_f^\kappa$. This result is based on our recent work \cite{MaxTheta1} where we have proven the maximality of $\mathsf{A}_2$ in $\mathcal{L}_2(1)$ for the centred lattice theta functions, i.e. $L\mapsto \theta_{L+c_L}(\alpha)$, where $c_L$ is the center of the unit cell $Q_L$ (see also Remark \ref{rmk:highalt}).

\begin{thm}[\textbf{Conservation of universal optimality - 2d shifted case}]\label{thm02}
Let $d=2$ and $\kappa=\{K, A_K, P_K\}$ be as in \eqref{eq:kappa} where $A_K\subset \R_+$, and be such that
\begin{equation}\label{eq:shift}
\forall k\in K, \forall i\in I_k, \quad \frac{p_{i,k}}{k}=c_L \textnormal{ modulo $Q_L$,}\quad L=\Z u_1 \oplus \Z u_2,\quad c_L:=\frac{u_1+u_2}{2},
\end{equation}
where $Q_L$ is the unit cell of $L$ defined by \eqref{eq:QL} with a Minkowski basis $\{u_1,u_2\}$ and its center $c_L$. Then, for all $f\in \mathcal{F}_2^{cm}$, $\mathsf{A}_2$ is the unique minimizer of $E^\kappa_f$ in $\mathcal{L}_2(1)$.
\end{thm}

\begin{example}
Theorem \ref{thm02} holds in a particularly simple case, when $k=2$ and $p_{i,2}=u_1+u_2\in L$.
\end{example}

\begin{remark}[Conjecture in dimensions $d\in \{8,24\}$]\label{rmk:conjoptalt}
Theorem \ref{thm02} is based on the fact that $\mathsf{A}_2$ has been shown to be the unique maximizer of $E_f^c$ defined in \eqref{eq:defaltercenter} in $\mathcal{L}_d(1)$ for any $f\in \mathcal{F}_d^{cm}$ (see also Remark \ref{rmk:highalt}). As discussed in \cite{MaxTheta1}, we believe that this result still holds in dimensions 8 and 24 for $\mathsf{E}_8$ and the Leech lattice $\Lambda_{24}$, as well as our Theorem \ref{thm02}.
\end{remark}

\begin{remark}[Phase transition for the minimizer in the Gaussian case - Numerical observation]
In the non-universally optimal case of Theorem \ref{thm0} and the shifted case satisfying \eqref{eq:shift}, numerical investigations suggest that the minimizer of $E_f^\kappa$ exhibits a phase transition as the density decreases.\\
\textit{Non-shifted case.} Let us consider the example $f(r)=e^{-\pi \alpha r}$ given in Theorem \ref{thm0} (i.e. $f(r^2)$ is a Gaussian function) and $f_{\kappa}(r)=e^{-\pi \alpha r}-0.1 e^{-2\pi \alpha r}$ (defined by \eqref{eq:fAK}), $\kappa:=\{2, 0.1,\emptyset\}$, corresponding to removing $a_2=0.1$ times the sublattice $2L$ ($k=2$) from the original lattice $L$. In dimension $d=2$, we numerically observe an interesting phase transition of type `triangular-rhombic-square-rectangular' for the minimizer of $E_f^{\kappa}$ in $\mathcal{L}_2(1)$ as $\alpha$ (which plays the role of the inverse density here) increases.\\
\textit{Shifted case with $a_k<0$.} Let us assume that $K=\{2\}$, $A_K:=\{a_2<0\}$, $I_2=\{1\}$ and $p_{1,2}=u_1+u_2$ in such a way that \eqref{eq:shift} is satisfied. If we consider $f(r)=e^{-\pi\alpha r}$, then for all the negative parameters $a_2$ we have chosen, the minimizer of $E_f^\kappa[L]:=\theta_L(\alpha)+|a_2| \theta_{L+c_L}(\alpha)$ in $\mathcal{L}_2(1)$ numerically shows the same phase transition of type `triangular-rhombic-square-rectangular' as $\alpha$ increases.\\
 This type of phase transition seems to have a certain universality in dimension 2 since it was also observed for Lennard-Jones energy \cite{Beterloc}, Morse energy \cite{LBMorse}, Madelung-like energies \cite{BFK20} and proved for 3-blocks copolymers \cite{LuoChenWei} and two-component Bose-Einstein condensates \cite{LuoWeiBEC} by Wei et al..
\end{remark}

\begin{remark}[Optimality of $\Z^d$ in the orthorhombic case]\label{rmk:ortho}
Another type of universal optimality is known in the set of orthorhombic lattices, i.e. the lattice $L$ which can be represented by an orthogonal basis. As proved by Montgomery in \cite[Thm. 2]{Mont}, the cubic lattice $\Z^d$ is universally minimal among orthorhombic lattices of unit density in any dimension (see also \cite[Rmk. 3.1]{BFK20}). The proof of Theorem \ref{thm0} can be easily adapted to show the same optimality result for $\Z^d$ among orthorhombic lattices of unit density. Furthermore, it has also been shown (see e.g. \cite[Prop. 1.4]{BeterminPetrache}) that $\Z^d$ is the unique maximum of $L\mapsto E_f[L+c_L]$ among orthorhombic lattices of fixed density for any $f\in \mathcal{F}_d^{cm}$. Therefore, the proof of Theorem \ref{thm02} can be also easily adapted in this orthorhombic case in order to show the universal optimality of $\Z^d$ in this particular shifted case. Moreover, all the next results involving any universally optimal lattice can be rephrased for $\Z^d$ in the space of orthorhombic lattices. Examples of applications of such result will be discussed in Section \ref{subsec:ions}.
\end{remark}


We now give a general criterion that ensures the conservation of an universal optimizer's minimality for $E_f^{\kappa}$.

\begin{thm}[\textbf{General criterion for minimality conservation - Non-shifted case}]\label{thm1}
Let $d\geq 2$, $\kappa=\{K, A_K,\emptyset\}$ be as in \eqref{eq:kappa} (possibly empty) where $A_K\subset \R_+$, and $L_d$ be universally optimal in $\mathcal{L}_d(1)$. Furthermore, let $f\in \mathcal{F}_d$ be such that $d\mu_f(t)=\rho_f(t)dt$ and $f_{\kappa}$ be defined by \eqref{eq:fAK}. Then:
\begin{enumerate}
\item For any $\kappa$, we have $f_\kappa(r)=\displaystyle \int_0^\infty e^{-rt}d\mu_{f_\kappa}(t)$ where 
$$
d\mu_{f_\kappa}(t)=\rho_{f_\kappa}(t)dt,\quad \rho_{f_\kappa}(t)=\rho_f(t)-\sum_{k\in K} \frac{a_k}{k^2}\rho_f\left( \frac{t}{k^2}\right).
$$
\item The following equivalence holds: $f_{\kappa}\in \mathcal{F}_d^{cm}$ if and only if 
\begin{equation}\label{eq:condthm}
\forall t>0,\quad \rho_f(t)\geq \sum_{k \in K}\frac{a_k}{k^2}\rho_f\left( \frac{t}{k^2} \right);
\end{equation}
\item If \eqref{eq:condthm} holds, then $L_d$ is the unique minimizer of $E_f^{\kappa}$ in $\mathcal{L}_d(1)$.
\item If there exists $V>0$ such that for a.e. $y\ge 1$ there holds 
\begin{equation}\label{eq:condgV}
g_V(y):=\rho_{f_\kappa}\left(  \frac{\pi y}{V^{\frac{2}{d}}}\right)+y^{\frac{d}{2}-2}\rho_{f_\kappa}\left( \frac{\pi}{V^{\frac{2}{d}} y} \right)\geq 0,
\end{equation}
then $V^{\frac{1}{d}}L_d$ is the unique minimizer of $E_f^\kappa$ in $\mathcal{L}_d(V)$.
\end{enumerate}
\end{thm}

The fourth point on Theorem \ref{thm1} generalizes our two-dimensional result  \cite[Thm. 1.1]{BetTheta15} to any dimension and with possible periodic arrays of defects. It is an important result since only few minimality results for $E_f$ are available for non-completely monotone potentials $f\in \mathcal{F}_d\backslash \mathcal{F}_d^{cm}$, and this also the first result of this kind for charged lattices (i.e. when the particles are not of the same kind). Condition \eqref{eq:condgV} has been used in dimension $d=2$ in \cite{BetTheta15,LBMorse} for proving the optimality of a triangular lattice at fixed density for non-convex sums of inverse power laws, differences of Yukawa potentials, Lennard-Jones potentials and Morse potentials and we expect the same property to hold in higher dimension. In Theorem \ref{thm3LJ}, we will give an example of such application in any dimension $d$ by applying the fourth point of Theorem \ref{thm1} to Lennard-Jones type potentials. We now add a very important remark concerning the adaptation of the fourth point of Theorem \ref{thm1} in the general periodic case, i.e. for crystallographic point packings (see \cite[Def. 2.5]{Aperiodic}).

\begin{remark}[\textbf{Crystallization at fixed density as a consequence of Cohn-Kumar Conjecture}]\label{rmk:CK}
When $\kappa=\emptyset$, i.e. all the particles are present and of the same kind, the proof of point 4. of Theorem \ref{thm1} admits a straightforward adaptation in the periodic case, i.e among all configurations $\mathcal{C}=\bigcup_{i=1}^N \left(\Lambda+v_k\right)\in \mathcal{S}$ being $\Lambda$-periodic of unit density, where $\Lambda\in \mathcal{L}_d$, i.e. such that $|\Lambda|=N$, and with a $f$-energy defined for $V>0$ by
$$
E_f[V^{\frac{1}{d}}\mathcal{C}]:= \frac{1}{N}\sum_{j,k=1}^N \sum_{x\in \Lambda \backslash \{v_k-v_j\}} f\left(V^{\frac{2}{d}} |x+v_k-v_j|^2\right).
$$
Using again the representation of $f$ as a superposition of Gaussians combined with the Jacobi transformation formula (see the proof of Theorem \ref{thm1}), the same condition \eqref{eq:condgV} ensures the crystallization on $L_d$ at fixed density once we know its universal optimality in the set of all periodic configurations with fixed density $V^{-1}$. This result is in the same spirit as the one derived by Petrache and Serfaty in \cite{PetracheSerfatyCryst} for Coulomb and Riesz interactions. In dimensions $d\in \{8,24\}$, \eqref{eq:condgV} implies the crystallization on $\mathsf{E}_8$ and $\Lambda_{24}$ at fixed density $V^{-1}$ as a consequence of \cite{CKMRV2Theta} whereas in dimension $d=2$ it is conjectured by Cohn and Kumar in \cite{CohnKumar} that the same holds on the triangular lattice. It is in particular true for the Lennard-Jones potential at high density as a simple application of our Theorem \ref{thm3LJ}.
\end{remark}

Using exactly the same arguments as the fourth point of Theorem \ref{thm1}, we show the following result which gives a sufficient condition on an interaction potential $f$ for a universal maximizer $L_d^\pm$ of $\theta_L^\pm(\alpha)$ to be optimal for $E_f^\pm$, where
\begin{equation}\label{eq:thetaEfalt}
\theta_L^\pm(\alpha):=\sum_{p\in L} \varphi_\pm(p) e^{-\pi \alpha |p|^2},\quad \textnormal{and}\quad E_f^\pm[L]:=\sum_{p\in L\backslash \{0\}} \varphi_\pm(p) f(|p|^2),
\end{equation}
with $L=\bigoplus_{i=1}^d \Z u_i$, $\{u_1,...,u_d\}$ being its Minkowski basis, and $\varphi_\pm(p)=\sum_{i=1}^d m_i$ for $p=\sum_{i=1}^d m_i u_i$, $m_i\in \Z$ for all $i$. Remark that $E_f^\pm=E_f^\kappa$ when $\kappa=\{2, \{2,....2\},\{u_1,...,u_d\}\}$, $L=\bigoplus_{i=1}^d \Z u_i$. In particular, it holds for the triangular lattice $\mathsf{A}_2$ as a simple application of our main result in \cite{MaxTheta1}.

\begin{thm}[\textbf{Maximality of a universal maximizer for $\theta_L^\pm$ - Shifted case}]\label{thm02max}
Let $d\geq 2$, $V>0$, $\kappa=\{2, \{2,....2\},\{u_1,...,u_d\}\}$, where a generic lattice is written $L=\bigoplus_{i=1}^d \Z u_i$, $\{u_1,...,u_d\}$ being its Minkowski basis, and $L_d^\pm$ be the unique maximizer of $\theta_L^\pm(\alpha)$, defined by \eqref{eq:thetaEfalt}, in $\mathcal{L}_d(1)$ and for all $\alpha>0$. If $f\in \mathcal{F}_d$ satisfies \eqref{eq:condgV}, then $V^{\frac{1}{d}}L_d^\pm$ is the unique maximizer of $E_f^\kappa$ (equivalently of $E_f^\pm$ defined by \eqref{eq:thetaEfalt}) in $\mathcal{L}_d(V)$.
\end{thm}

\begin{remark}[Adaptation to shifted $f$-energy]\label{rmk:highalt}
We believe that Theorem \ref{thm02max} also holds for $\mathsf{E}_8$ and $\Lambda_{24}$ (see \cite[Conj. 1.3]{MaxTheta1} and Remark \ref{rmk:conjoptalt}). Furthermore, the same kind of optimality result could be easily derived for any energy shifted energy of type $L\mapsto E_f[L+c]$ where $c\in Q_L$ is fixed as a function of the vectors in the Minkowski basis $\{u_i\}$ and when one knows a universal minimizer or maximizer for $L\mapsto E_f[L+c]$, $f\in \mathcal{F}_d^{cm}$. However, no other result concerning any optimality of a lattice for such kind of energy is currently available when $c\not\in\{L, c_L\}$.
\end{remark}

The rest of our results are all given in the non-shifted case $P_K=\emptyset$. It is indeed a rather difficult task to minimize the sum of shifted and/or non-shifted energies of type $E_f$. Very few results are available and the recent work by Luo and Wei \cite{LuoWeiBEC} has shown the extreme difficulty to obtain any general result for completely monotone function $f$. Shifting the lattices by a non-lattice point which is not the center $c_L$ appears to be deeply more tricky in terms of energy optimization.

\medskip

We remark that, since $\mathcal{F}_d^{cm}$ is not stable by difference, it is not totally surprising that Theorem \ref{thm0} holds. Furthermore, identifying the largest space of all functions $f$ such that $E_f$ is uniquely minimized by $L_d$ in $\mathcal{L}_d(1)$ seems to be very challenging (see \cite{OptinonCM}). Therefore a natural question in order to identify a large class of potentials $f$ such that the minimality of an universal optimizer $L_d$ holds for $E_f^{\kappa}$ is the following: what are the completely monotone potentials $f\in \mathcal{F}_d^{cm}$ satisfying \eqref{eq:condthm}, i.e. such that $f_{\kappa}\in \mathcal{F}_d^{cm}$? The following corollary of Theorem \ref{thm1} gives an example of such potentials, where we define, for $s>0$ and any $A_K=\{a_k\}_{k\in K}$, $K\subset \N\backslash \{1\}$,
\begin{equation}\label{eq:zetaAK}
\mathsf{L}(A_K,s):=\sum_{k\in K} \frac{a_k}{k^s}.
\end{equation}
Notice that the notation of \eqref{eq:zetaAK} is inspired by the one of Dirichlet L-series that are generalizing the Riemann zeta function (see e.g. \cite[Chap. 10]{Cohen}). For us, the arithmetic function appearing in a Dirichlet series is simply replaced by $A_K$ and can be finite.

\begin{corollary}[\textbf{Minimality conservation for special $f$ - Non-shifted case}]\label{Cor1}
Let $d\geq 2$ and $f\in \mathcal{F}_d^{cm}$ be such that $d\mu_f(t)=\rho_f(t)dt$ and $\rho_f$ be an increasing function on $\R_+$. Let $\kappa=\{K,A_K,\emptyset\}$ be as in \eqref{eq:kappa} where $A_K=\{a_k\}_{k\in K}\subset \R_+$ and be such that $\mathsf{L}(A_K,s)$ defined by \eqref{eq:zetaAK} satisfies $\mathsf{L}(A_K,2)\leq 1$. If $L_d$ is universally optimal in $\mathcal{L}_d(1)$, then $L_d$ is the unique minimizer of $E_f^{\kappa}$ in $\mathcal{L}_d(1)$.
\end{corollary}

\begin{example}[Potentials satisfying the assumptions of Corollary \ref{Cor1}]\label{ex:ipyuka}
There are many examples of potentials $f$ such that Corollary \ref{Cor1} holds. For instance, this is the case for the parametrized potential $f=f_{\sigma,s}$ defined for all $r>0$ by $f_{\sigma,s}(r)= \frac{e^{-\sigma r}}{r^{s}}$, $\sigma>0$, $s>1$, since $d\mu_{f_{\sigma,s}}(t)=\frac{(t-\sigma)^{s-1}}{\Gamma(s)}\mathds{1}_{[\sigma,\infty)}(t)dt$ and $t\mapsto \frac{(t-\sigma)^{s-1}}{\Gamma(s)}\mathds{1}_{[\sigma,\infty)}(t)$ are increasing functions on $\R_+$. Notice that the inverse power law $f(r)=r^{-s}$ with exponent $s>d/2\geq 1$ (if $\sigma=0$) and the Yukawa potential $f(r)=e^{-\sigma r} r^{-1}$ with parameter $\sigma>0$ (if $s=1$) are special cases of $f_{\sigma,s}$.
\end{example}

\subsection{The inverse power law and Lennard-Jones cases}

In this subsection, we restrict our study to combinations of inverse power laws, since they are the building blocks of many interaction potentials used in molecular simulations (see e.g. \cite{Kaplan}). Their homogeneity simplifies a lot the energy computations and allows us to give a complete picture of the periodic arrays of defects effects with respect to the values of $\mathsf{L}$ defined by \eqref{eq:zetaAK}.

\medskip

In the following result, we show that the values of $\mathsf{L}(A_K, 2s)$ plays a fundamental role in the minimization of $E_f^{\kappa}$ when $f$ is an inverse power law.

\begin{thm}[\textbf{The inverse power law case - Non-shifted case}]\label{thm2ip}
Let $d\geq 2$ and $f(r)=r^{-s}$ where $s>d/2$. Let $\kappa=\{K, A_K,\emptyset\}$ be as in \eqref{eq:kappa} and be such that $\mathsf{L}(A_K, 2s)$ defined by \eqref{eq:zetaAK} is absolutely convergent. We have:
\begin{enumerate}
\item If $\mathsf{L}(A_K, 2s)<1$, then $L_0$ is a minimizer of $L\mapsto \zeta_L(2s)$ in $\mathcal{L}_d(1)$ if and only if $L_0$ is a minimizer of $E_f^{\kappa}$ in $\mathcal{L}_d(1)$.
\item  If $\mathsf{L}( A_K, 2s)>1$, then $L_0$ is a minimizer of $L\mapsto \zeta_L(2s)$ in $\mathcal{L}_d(1)$ if and only if $L_0$ is a maximizer of $E_f^{\kappa}$ in $\mathcal{L}_d(1)$.
\end{enumerate}
In particular, for any $K\subset \N\backslash \{1\}$, if $a_k=1$ for all $k \in K$, then $L\mapsto \zeta_L(2s)$ and $E_f^\kappa$ have the same minimizers in $\mathcal{L}_d(1)$.
\end{thm}

\begin{examples}[Minimizers of the Epstein zeta function]
In dimensions $d\in \{2,8,24\}$, the minimizer $L_0$ of $L\mapsto \zeta_L(2s)$ in $\mathcal{L}_d(1)$ is, respectively, $\mathsf{A}_2$, $\mathsf{E}_8$ and $\Lambda_{24}$ as consequences of \cite{Mont,CKMRV2Theta}. In dimension $d=3$, Sarnak and Str\"ombergsson have conjectured in \cite[Eq. (44)]{SarStromb} that the face-centred cubic lattice $\mathsf{D}_3$ (see Fig. \ref{fig:Lattices}) is the unique minimizer of $L\mapsto \zeta_L(2s)$ in $\mathcal{L}_3(1)$ if $s>3/2$.
\end{examples}

Many applications of point 4. of Theorem \ref{thm1} can then be shown for non-convex sums of inverse power laws, differences of Yukawa potentials or Morse potentials by following the lines of \cite{BetTheta15}. In this paper, we have chosen to focus on Lennard-Jones type potentials since it is possible to have a complete description of the effect of non-shifted periodic arrays of vacancies using the homogeneity of the Epstein zeta functions. It is also known that Lennard-Jones type potentials play an important role in molecular simulation (see e.g. \cite[Sect. 6.3]{BetTheta15} and \cite[Sect. 5.1.2]{Kaplan}).

\medskip

In our last results, we define the Lennard-Jones type potential by
\begin{equation}\label{eq:LJdef}
f(r)=\frac{c_2}{r^{x_2}}-\frac{c_1}{r^{x_1}}\quad \textnormal{where}\quad (c_1,c_2)\in (0,\infty), \quad x_2>x_1>d/2,
\end{equation}
which is a prototypical example of function where $\mu_f$ is not nonnegative everywhere, and a difference of completely monotone functions. We discuss the optimality of a universally optimal lattice $L_d$ for $E_f^{\kappa}$ with respect to the values of $\mathsf{L}( A_K, 2 x_i)$, $i\in \{1,2\}$ as well as the shape of the global minimizer of $E_f^{\kappa}$, i.e. its equivalence class in $\mathcal{L}_d$ modulo rotation and dilation (as previously defined in \cite{OptinonCM}).

\begin{thm}[\textbf{The Lennard-Jones case - Non-shifted case}]\label{thm3LJ}
Let $d\geq 2$, $f$ be defined by \eqref{eq:LJdef} and $\kappa=\{K,A_K,\emptyset\}$ be as in \eqref{eq:kappa} (possibly empty) and be such that $\mathsf{L}( A_K, 2x_i)$, $i\in \{1,2\}$ defined by \eqref{eq:zetaAK} are absolutely convergent. Let $L_d$ be universally optimal in $\mathcal{L}_d(1)$. Then:
\begin{enumerate}
\item If $\mathsf{L}(A_K, 2x_2)<\mathsf{L}( A_K, 2x_1)<1$, then for all $V>0$ such that
$$
V\leq V_\kappa:=\pi^{\frac{d}{2}}\left(  \frac{c_2(1-\mathsf{L}( A_K, 2x_2))\Gamma(x_1)}{c_1(1-\mathsf{L}( A_K, 2x_1))\Gamma(x_2)}\right)^{\frac{d}{2(x_2-x_1)}},
$$
the lattice $V^{\frac{1}{d}}L_d$ is the unique minimizer of $E_f^{\kappa}$ in $\mathcal{L}_d(V)$ and there exists $V_1>0$ such that it is not a minimizer of $E_f^{\kappa}$ for $V>V_1$. Furthermore, the shape of the minimizer of $E_f$ and $E_f^{\kappa}$ are the same in $\mathcal{L}_d$.
\item If $\mathsf{L}(A_K,2x_1)>\mathsf{L}( A_K, 2x_2)>1$, then $E_f^{\kappa}$ does not have any minimizer in $\mathcal{L}_d$ and for all $V<V_\kappa$, $V^{\frac{1}{d}}L_d$ is the unique maximizer of $E_f^{\kappa}$ in $\mathcal{L}_d(V)$.
\item If $\mathsf{L}( A_K, 2x_1)>1>\mathsf{L}( A_K, 2x_2)$, then $E_f^{\kappa}$ does not have any minimizer in $\mathcal{L}_d$ but $V^{\frac{1}{d}}L_d$ is the unique minimizer of $E_f^{\kappa}$ in $\mathcal{L}_d(V)$ for all $V>0$.
\end{enumerate}
\end{thm}

\begin{remark}[Increasing of the threshold value $V_\kappa$]
The fact that $1-\mathsf{L}( A_K, 2x_2)> 1-\mathsf{L}( A_K, 2x_1)$ implies that the threshold value $V_\kappa$ is larger in the $\kappa \neq \emptyset$ case than in the case without defect $\kappa=\emptyset$. The same is expected to be true for any non-convex sum of inverse power law with a positive main term as $r\to 0$ (see \cite[Prop. 6.4]{BetTheta15} for a twp-dimensional example in the no-defect case $\kappa=\emptyset$). It is also totally straightforward to show that $V_\kappa \to V_\emptyset$ as $\min K$ tend to $+\infty$.
\end{remark}

\begin{remark}[Global minimality of $\mathsf{A}_2$ among lattices for Lennard-Jones type potentials]\label{rmk:globA2}
In dimension $d=2$, the triangular lattice $L_2=\mathsf{A}_2$ has been shown in \cite[Thm. 1.2.2]{BetTheta15} to be the shape of the global minimizer of $E_f$ in $\mathcal{L}_2$ when $\pi^{-x_2}\Gamma(x_2)x_2<\pi^{-x_1}\Gamma(x_1)x_1$. Point 1. of Theorem \ref{thm3LJ} implies that the same holds when $\mathsf{L}( A_K, 2x_2)<\mathsf{L}( A_K, 2x_1)<1$.
\end{remark}

\subsection{Conclusion}

From all our results, we conclude that is possible to remove or substitute several infinite periodic sets of points from all the lattices (i.e. an integer sublattices) and to conserve the already existing minimality properties, but only in a certain class of potentials or sublattices. Physically, it means that adding point defects to a crystal can be without any effect on its ground state if we assume the interaction between atoms to be well-approximate by a pairwise potential (Born model \cite{Tosi}) and the sublattices to satisfy some simple properties. We give several examples in Section \ref{subsec:examples} and our result are the first known general results giving global optimality of ionic crystals. In particular, the Kagome lattice (see Figure \ref{fig:introionic}) is shown to be the global minimizer for the interaction energies discussed in this paper in the class of (potentially shifted) lattices $L\backslash 2L$ where $L\in \mathcal{L}_2(1)$. This is, as far as we know, the first results of this kind for the Kagome lattice. We also believe that the results and techniques derived in this paper can be applied to other ionic crystals and other general periodic systems. 

\medskip

Furthermore, this paper also shows the possibility to check the optimality of a structure while 'forgetting' many points which, in a certain sense, do not play any role (vacancy case). This allow to simplify both numerical investigations -- leading to a shorter computational time -- and mathematical estimates for these energies. We voluntarily did not explore further this fact since it is only relevant in low dimensions because the computational time of such lattice sums is exponentially growing and gives unreachable durations in dimension $d\geq 4$ for computing many values of the energies, especially in dimensions $d\in \{8,24\}$ where our global optimality results are applicable. 

\medskip

In dimension $d=3$, i.e. where the everyday life real crystals exist, our results only apply -- combined with the one from \cite{Beterminlocal3d} -- to the conservation of local minimality in the cubic lattices cases ($\Z^3$, $\mathsf{D}_3$ and $\mathsf{D}_3^*$) for the Epstein zeta function, the lattice theta function and the Lennard-Jones type energies. We believe that our result will find other very interesting applications in dimension 3 once global optimality properties will be shown for the lattice theta functions and the Epstein zeta functions (Sarnak-Str\"ombergsson conjectures \cite{SarStromb}).

\medskip

Even though the inverse power laws and Lennard-Jones cases have been completely solved here, we still ignore what is the optimal result that holds for ensuring the robustness of the universal optimality among lattices. An interesting problem would be to find a necessary condition for this robustness. Furthermore, we can also ask the following question: is it enough to study this kind of minimization problem in a (small) ball centred at the origin? In other words: can we remove all the points that are far enough from $O$ and conserving the minimality results? Numerical investigations and Figure \ref{fig:squarecharges} tend to confirm this fact, and a rigorous proof of such property would deeply simplify the analysis of such lattice energies.

\section{Applications: The Kagome lattice and other ionic structures}\label{subsec:examples}

We now give several examples of applications of our results. In particular, we identify interesting structures that are minimizers of $E_f$ in classes of sparse and charged lattices.

\subsection{The Kagome lattice}

Being the vertices of a trihexagonal tiling, this structure -- wich is actually not a lattice as we defined it in this paper -- that we will write $\mathsf{K}:=\mathsf{A}_2\backslash 2 \mathsf{A}_2$ is the difference of two triangular lattices of scale ratio $2$ (see Fig. \ref{fig:Kagome}). Some minerals -- which display novel physical properties connected with geometrically frustrated magnetism -- like jarosites and herbertsmithite contain layers having this structure (see \cite{Kagome} and references therein). We can therefore apply our results of Section \ref{subsec:main} with $\kappa=\{2,1,\emptyset\}$ or $\kappa=\{2,1,u_1+u_2 \}$. The following optimality results for $E_f$ in the class of lattices of the form $L\backslash 2L$ (or $L\backslash (2L + u_1+u_2)$ in the shifted case) are simple consequences of our results combined with the universal optimality of $\mathsf{A}_2$ among lattices proved by Montgomery in \cite{Mont}:
\begin{enumerate}
\item \textit{Universal optimality of $\mathsf{K}$.} Applying Theorem \ref{thm02} to $\kappa = \{2,1, u_1+u_2  \}$, it follows that for all $f\in \mathcal{F}_2^{cm}$, the shifted Kagome lattice $\mathsf{K}+(1/2,-\sqrt{3}/2)$ (see Fig. \ref{fig:Kagome}) is the unique minimizer of $E_f$ among lattices of the form $L\backslash (2L + u_1+u_2)$, where $L=\Z u_1 \oplus \Z u_2\in \mathcal{L}_2(1)$.
\item \textit{Minimality of $\mathsf{K}$ at all densities for certain completely monotone potentials.}  A direct consequence of Theorem \ref{thm1} is the following. For any completely monotone function $f\in \mathcal{F}_2^{cm}$ such that $d\mu_f(t)=\rho_f(t)dt$ and $\rho_f$ is an increasing function, the Kagome lattice $\mathsf{K}$ is the unique minimizer of $E_f$ among all the two-dimensional sparse lattices $L\backslash 2 L$ where $L\in \mathcal{L}_2(1)$. This is the case for instance for $f=f_{\sigma,s}$ defined in Example \ref{ex:ipyuka}, including the inverse power laws and the Yukawa potential.
\item \textit{Optimality at high density for Lennard-Jones interactions.} Applying Theorem \ref{thm3LJ}, we obtain its optimality at high density: if $f(r)=c_2r^{-x_2}-c_1 r^{-x_1}$, $x_2>x_1>1$ is a Lennard-Jones potential, then the unique minimizer of $E_f$ at high density among all the two-dimensional sparse lattices $L\backslash 2 L$, where $L$ has fixed density, has the shape of $\mathsf{K}$.
\item \textit{Global optimality for Lennard-Jones interactions with small exponents.} Furthermore, using Theorem \ref{thm3LJ} and \cite[Thm. 1.2.2]{BetTheta15} (see also Remark \ref{rmk:globA2}), we obtain the following interesting result in the Lennard-Jones potential case: if $\pi^{-x_2}\Gamma(x_2)x_2<\pi^{-x_1}\Gamma(x_1)x_1$, then the unique global minimizer of $E_f$ among all the possible sparse lattices $L\backslash 2 L$ has the shape of $\mathsf{K}$. 
\end{enumerate}

\begin{figure}[!h]
\centering
\includegraphics[width=5cm]{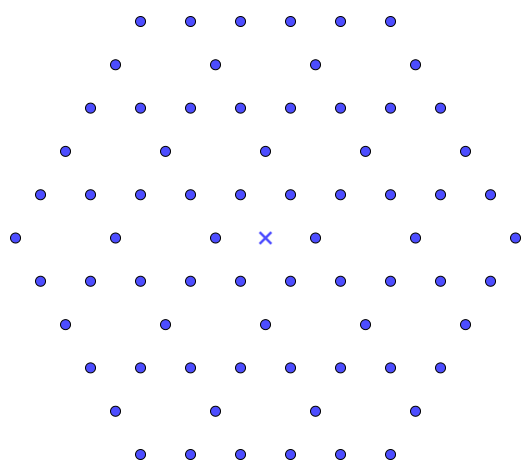} \qquad \includegraphics[width=5cm]{kagomeshift}
\caption{\small{Two patches of the Kagome lattice. On the left, the origin $O$ does not belong to $\mathsf{K}$ and is the center of one of the hexagons. On the right, $O$ belongs to a shifted version $\mathsf{K}+(1/2,-\sqrt{3}/2)$.}}
\label{fig:Kagome}
\end{figure}

These are the first minimality results for $\mathsf{K}$ in a class of periodic configurations. We recall that a non-optimality result has also been derived by Grivopoulos in \cite{Grivopoulos} for Lennard-Jones potential in the case of free particles, and different attempts have been made for obtaining numerically or experimentally a Kagome structure as an energy ground state (see e.g. \cite{HyunTorquato,Kagomecolloids,TorquatoKagome}).

\begin{remark}[The honeycomb lattice]
We notice that the honeycomb lattice $\mathsf{H}:=\mathsf{A}_2\backslash \sqrt{3}\mathsf{A}_2$, also constructed from the triangular lattice, does not belong to the set of sparse lattices $L\backslash k L$, $k\in \N$. That is why no optimality result for $\mathsf{H}$ is included in this paper.
\end{remark}

\subsection{Rock-salt vs. other ionic structures}\label{subsec:ions}
We recall that, in \cite{MaxTheta1}, we have shown with Faulhuber the universal optimality of the triangular lattice among lattices with alternating charges, i.e. the fact that $\mathsf{A}_2$ uniquely maximizes 
\begin{equation}\label{eq:madelung}
L\mapsto\theta_L^\pm(\alpha):=\sum_{p\in L} \varphi_\pm(p) e^{-\pi \alpha |p|^2} \quad \textnormal{and}\quad \zeta_L^\pm(s):=\sum_{p\in L\backslash \{0\}} \frac{\varphi_\pm(p)}{|p|^s},\quad L=\Z u_1\oplus \Z u_2,
\end{equation}
in $\mathcal{L}_2(1)$, where, for all $p=m u_1+n u_2$, $\varphi_\pm(p):=m + n $. Notice that the maximality result at all scales for the alternating lattice theta function is equivalent with the fact that $\mathsf{A}_2$ maximizes 
$$
L\mapsto E_f^\kappa[L]:=E_f[L]-2E_f[2L+u_1]-2E_f[2L+u_2],\quad \textnormal{where}\quad \kappa:=\{2, \{2,2\}, \{u_1,u_2\}\}
$$
in $\mathcal{L}_2(1)$ for any $f\in \mathcal{F}_2^{cm}$. It has been also proven in \cite[Thm. 1.4]{BeterminPetrache} that $\Z^d$ is the unique maximizer of the $d$-dimensional generalization of the two lattice energies $\theta_L^\pm(\alpha)$ and $\zeta_L^\pm(s)$ among $d$-dimensional orthorhombic (rectangular) lattices of fixed unit density, whereas it is a minimizer of the lattice theta functions and the Epstein zeta functions defined in \eqref{eq:zetatheta}.
Furthermore, applying Theorem \ref{thm1} in dimension $d=2$ (resp. any $d$), we see that $\mathsf{A}_2$ (resp. $\Z^d$) minimizes in $\mathcal{L}_2(1)$ (resp. among the orthorhombic lattices of unit density) the energy
\begin{equation}\label{eq:energyapplication}
E_f^{\kappa}[L]:= \zeta_L(s)-2\zeta_{kL}(s),\quad f(r)=r^{-s},\quad K=\{k\},\quad a_k=2,
\end{equation}
for all $s>d/2$. We remark that $\Z^d$, $d\in \{2,3\}$ is also a saddle point (see \cite{Mont,Beterminlocal3d}) of $E_f^{\kappa}$ in $\mathcal{L}_d(1)$. It is then interesting to see how the array of substitutional defects with charges $-1$ plays a totally different role for this energy (see also Fig. \ref{fig:squarecharges} and Fig. \ref{fig:tricharges}). This seems to confirm that the role of the nearest-neighbors of the origin is fundamental, since they are actually the main terms of the energy when the potential is decreasing fast at infinity.

\begin{figure}[!h]
\centering
\includegraphics[width=4cm]{grid_salt} \qquad \includegraphics[width=4cm]{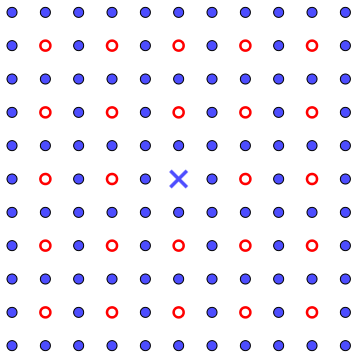}\qquad \includegraphics[width=4.2cm]{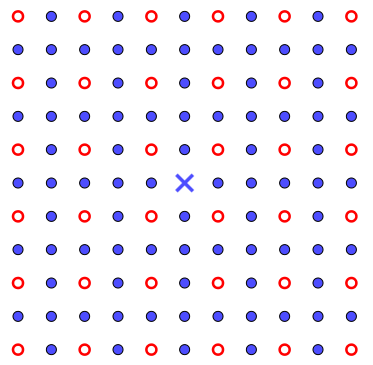}
\caption{\small{Three periodic arrays of defects on $\Z^2$. Blue points $\textcolor{blue}{\bullet}$ are points with charges $+1$ and red points $\textcolor{red}{\circ}$ are with charges $-1$. For the inverse power laws energies, the left-hand configuration is the unique maximizer among rectangular lattices of fixed density with alternation of charges whereas the centred configuration is its unique minimizer with this distribution of charges among rectangular lattices. However, the configuration on the right is a saddle point of any energy on the form $E_f$, $f\in \mathcal{F}_2^{cm}$ in this class of charged configurations. For the two structures on the left, the same is true in higher dimension while generalizing the ionic-like distribution on orthorhombic lattices.}}
\label{fig:squarecharges}
\end{figure}

\begin{figure}[!h]
\centering
\includegraphics[width=5cm]{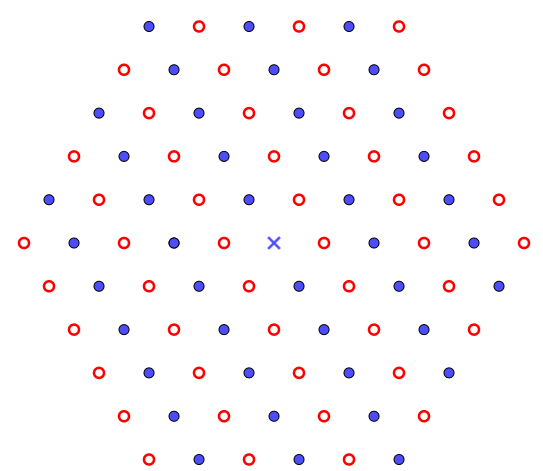} \quad \includegraphics[width=5cm]{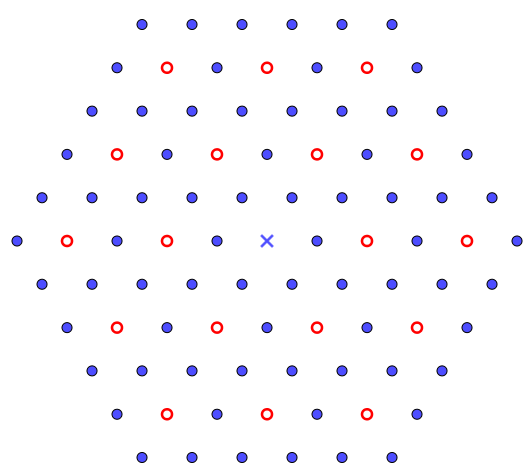} \quad \includegraphics[width=5cm]{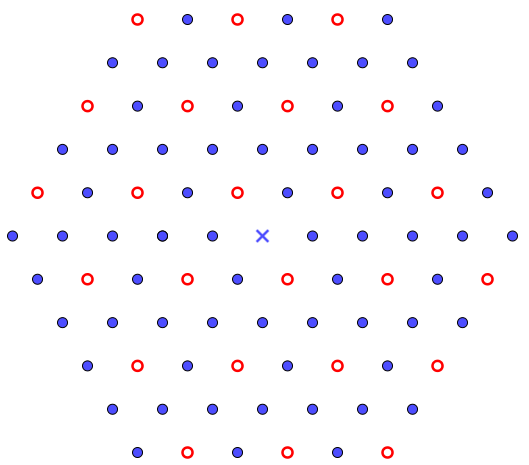}
\caption{\small{Three periodic arrays of defects on a patch of $\mathsf{A}_2$. Blue points $\textcolor{blue}{\bullet}$ are points with charges $+1$ and red points $\textcolor{red}{\circ}$ are with charges $-1$. On the left, the triangular alternate configuration maximizes $\zeta_L^\pm(s)$ in $\mathcal{L}_2(1)$ with this alternation of charges, while the configuration in the middle minimizes the inverse power law energy in this class of charged lattices. The configuration on the right minimizes any energy on the form $E_f$, $f\in \mathcal{F}_2^{cm}$ in this class of charged configurations.}}
\label{fig:tricharges}
\end{figure}

\section{Proofs of the main results}\label{sec:proof}
We first show Theorem \ref{thm0}, i.e. the non-robustness of universal optimality results under non-shifted periodic arrays of defects.

\begin{proof}[\textbf{Proof of Theorem \ref{thm0}}]
Let $\Lambda\in \{\mathsf{A}_2, \mathsf{E}_8,\Lambda_{24}\}$. We consider the potential $f(r):=e^{-\pi \alpha r}$ where $\alpha>0$. For all $k\in \N\backslash \{1\}$, all $a_k>0$ and all $L\in \mathcal{L}_d(1)$, we have, using the fact that $\theta_{k L}(\alpha)=\theta_L(k^2 \alpha)$,
$$
E_f^{\kappa}[L]=\theta_L(\alpha)-a_k\theta(k^2 \alpha).
$$
Let us show that there exists $\alpha_d$ such that for all $0<\alpha<\alpha_d$, $\Lambda$ does not minimize $E_f^{\kappa}$ in  $\mathcal{L}_d(1)$.  Indeed, we have the following equivalence: for all $L\in \mathcal{L}_d(1)\backslash \{\Lambda\}$, $E_f^{\kappa}[L]>E_f^{\kappa}[\Lambda]$ if and only if
\begin{equation}\label{eq:equiv}
\inf _{L\in \mathcal{L}_d(1) \atop L\neq \Lambda}\frac{\theta_L(\alpha)-\theta_{\Lambda}(\alpha)}{\theta_L(k^2 \alpha)-\theta_{\Lambda}(k^2 \alpha)}>a_k.
\end{equation}
Let us show that \eqref{eq:equiv} does not hold for small $\alpha$, and in particular that the left term tends to $0$ as $\alpha\to 0$. We use Coulangeon and Sch\"urmann's work \cite[Eq. (21)]{Coulangeon:2010uq}, in the lattice case, who derived the Taylor expansion of the theta function as $L\mapsto \Lambda$ in $\mathcal{L}_d(1)$. We then obtain
\begin{align*}
\lim_{L\mapsto \Lambda\atop L\neq \Lambda} \frac{\theta_L(\alpha)-\theta_{\Lambda}(\alpha)}{\theta_L(k^2 \alpha)-\theta_{\Lambda}(k^2 \alpha)}&=\frac{\displaystyle \sum_{p\in \Lambda\backslash \{0\}}\pi \alpha |p|^2\left(\pi \alpha |p|^2 -2\right) e^{-\pi \alpha |p|^2}}{\displaystyle \sum_{p\in \Lambda\backslash \{0\}}\pi \alpha k^2 |p|^2\left(\pi \alpha k^2 |p|^2 -2\right) e^{-\pi \alpha k^2 |p|^2}}\\
&=k^{-2}\frac{\displaystyle  \sum_{p\in \Lambda\backslash \{0\}}\pi \alpha |p|^4 e^{-\pi \alpha  |p|^2}-2\sum_{p\in \Lambda\backslash \{0\}} |p|^2 e^{-\pi \alpha  |p|^2}}{\displaystyle  \sum_{p\in \Lambda\backslash \{0\}}\pi \alpha k^2 |p|^4 e^{-\pi \alpha k^2  |p|^2}-2\sum_{p\in \Lambda\backslash \{0\}} |p|^2 e^{-\pi \alpha k^2  |p|^2}}.
\end{align*}
By absolute convergence, the first term of both numerator and denominator are vanishing as $\alpha \to 0$. We therefore obtain that
$$
\lim_{\alpha \to 0} \lim_{L\mapsto \Lambda\atop L\neq \Lambda} \frac{\theta_L(\alpha)-\theta_{\Lambda}(\alpha)}{\theta_L(k^2 \alpha)-\theta_{\Lambda}(k^2 \alpha)}=\lim_{\alpha\to 0} k^{-2}\frac{\displaystyle \sum_{p\in \Lambda\backslash \{0\}} |p|^2 e^{-\pi \alpha  |p|^2}}{\displaystyle \sum_{p\in \Lambda\backslash \{0\}} |p|^2 e^{-\pi \alpha k^2  |p|^2}}=0,
$$
by comparing the convergence rate of these two exponential sums that are going to $+\infty$ as $\alpha\to 0$. It follows that \eqref{eq:equiv} does not hold for $\alpha<\alpha_d$ where $\alpha_d$ depends on $d$, $k$ and $a_k$, and the proof of the first part of the theorem is completed.

\medskip

The second part of the theorem is a simple consequence of the fact that $f_\kappa$ defined by \eqref{eq:fAK} belongs to $\mathcal{F}_d^{cm}$ if $f\in\mathcal{F}_d^{cm}$ and $a_k<0$ for all $k\in K$.
\end{proof}

The proof of our second result, namely Theorem \ref{thm02}, is a direct and simple consequence of our work \cite{MaxTheta1}.
\begin{proof}[\textbf{Proof of Theorem \ref{thm02}}]
If $p_{i,k}/k=c_L$ modulo $Q_L$ for all $k\in K$ and all $i\in I_k$, we obtain
$$
E_f^\kappa[L]=E_f[L]-\sum_{k\in K} a_k \sum_{i\in I_k}\sum_{p\in L\backslash \{0\}} f\left( k^2 \left| \frac{p_{i,k}}{k} + p \right|^2\right)=E_f[L]-\sum_{k\in K} a_k \sharp L_k E_{f(k^2 \cdot)}[L+c_L].
$$
As proved in \cite{MaxTheta1}, for any $f\in \mathcal{F}_2^{cm}$, $\mathsf{A}_2$ is the unique maximizer of $L\mapsto E_f[L+c_L]$ in $\mathcal{L}_2(1)$. It follows that $\mathsf{A}_2$, which uniquely minimizes $E_f$ in $\mathcal{L}_2(1)$ is the unique minimizer of $E_f^\kappa$ in $\mathcal{L}_2(1)$ since $a_k>0$ for all $k\in K$.
\end{proof}

We now show Theorem \ref{thm1} which gives a simple criterion for the conservation of the minimality of a universal optimizer.

\begin{proof}[\textbf{Proof of Theorem \ref{thm1}}]
In order to show the three first points, it is sufficient to show the first point of our theorem, i.e. the fact that $d\mu_{f_{\kappa}}(t)=\left(\rho_f(t)-\sum_{k \in K}a_k k^{-2}\rho_f\left(\frac{t}{k^2}\right)\right)dt$. We remark that $\rho_f$ is the inverse Laplace transform of $f$, i.e. $\rho_f(t)=\mathcal{L}^{-1}[f](t)$. By linearity, it follows that
$$
d\mu_{f_{\kappa}}(t)=\rho_{f_{\kappa}}(t)dt,\quad \textnormal{where}\quad \rho_{f_{\kappa}}(t)=\rho_f(t)-\sum_{k\in K} a_k \mathcal{L}^{-1}[f(k^2\cdot)](t).
$$
By the basic properties of the inverse Laplace transform, we obtain that, for all $t>0$, 
$$
\mathcal{L}^{-1}[f(k^2\cdot)](t)=k^{-2}\mathcal{L}^{-1}[f](k^{-2}t)=k^{-2}\rho_f(k^{-2}t),
$$ 
and our result follows by the universal optimality of $L_{d}$ in $\mathcal{L}_d(1)$ and the definition of completely monotone function.

\medskip

To show the last point of our theorem, we adapt \cite[Thm 1.1]{BetTheta15}. Let $L\in \mathcal{L}_d(1)$ and $V>0$, then we have
\begin{align}
E_f^\kappa[V^{\frac{1}{d}}L]&=\sum_{p\in L\backslash \{0\}} f_\kappa\left(V^{\frac{2}{d}}|p|^2\right)=\int_0^\infty \left[ \theta_L\left( \frac{V^{\frac{2}{d}} t}{\pi} \right)-1\right]\rho_{f_\kappa}(t)dt\nonumber\\
&=\frac{\pi}{V^{\frac{2}{d}}}\int_0^\infty \left[ \theta_L(y)-1 \right]\rho_{f_\kappa}\left( \frac{\pi y}{V^{\frac{2}{d}}} \right)dy\nonumber\\
&=\frac{\pi}{V^{\frac{2}{d}}}\int_0^1 \left[ \theta_L(y)-1 \right]\rho_{f_\kappa}\left( \frac{\pi y}{V^{\frac{2}{d}}} \right)dy+\frac{\pi}{V^{\frac{2}{d}}}\int_1^\infty \left[ \theta_L(y)-1 \right]\rho_{f_\kappa}\left( \frac{\pi y}{V^{\frac{2}{d}}} \right)dy\nonumber\\
&=\frac{\pi}{V^{\frac{2}{d}}}\int_1^\infty\left[ \theta_L\left( \frac{1}{y} \right)-1\right]\rho_{f_\kappa}\left(\frac{\pi}{y V^{\frac{2}{d}}}  \right)y^{-2}dy + \frac{\pi}{V^{\frac{2}{d}}}\int_1^\infty \left[ \theta_L(y)-1 \right]\rho_{f_\kappa}\left( \frac{\pi y}{V^{\frac{2}{d}}} \right)dy. \label{rewrite1}
\end{align}
A simple consequence of the Poisson summation formula is the well-known identity (see e.g. \cite[Eq. (43)]{ConSloanPacking})
\begin{equation}\label{jacobi}
\forall y>0,\quad \theta_L\left(\frac{1}{y} \right)=y^{\frac{d}{2}}\theta_{L^*}(y).
\end{equation}
From \eqref{jacobi}, we see that if $L_d$ is the unique minimizer of $L\mapsto \theta_L(\alpha)$ for all $\alpha>0,L\in\mathcal{L}_d(1)$ then $L_d^*=L_d$. From \eqref{rewrite1} and \eqref{jacobi}, for all $V>0, L\in \mathcal{L}_d(1)$, we have
\begin{align}\label{rewrite2}
E_f^\kappa[V^{\frac{1}{d}}L]&=\frac{\pi}{V^{\frac{2}{d}}}\int_1^\infty\left[ y^{\frac{d}{2}}\theta_{L^*}\left( y\right)-1\right]\rho_{f_\kappa}\left(\frac{\pi}{y V^{\frac{2}{d}}}  \right)y^{-2}dy + \frac{\pi}{V^{\frac{2}{d}}}\int_1^\infty \left[ \theta_L(y)-1 \right]\rho_{f_\kappa}\left( \frac{\pi y}{V^{\frac{2}{d}}} \right)dy.
\end{align}
and 
\begin{align}\label{rewrite3}
E_f^\kappa[V^{\frac{1}{d}}L]-E_f^\kappa[V^{\frac{1}{d}}L_d]&=\frac{\pi}{V^{\frac{2}{d}}}\int_1^\infty\left[ \theta_{L^*}\left( y\right)-\theta_{L_d}(y)\right]\rho_{f_\kappa}\left(\frac{\pi}{y V^{\frac{2}{d}}}  \right)y^{\frac{d}{2}-2}dy \nonumber\\
&\quad \quad + \frac{\pi}{V^{\frac{2}{d}}}\int_1^\infty \left[ \theta_L(y)-\theta_{L_d}(y) \right]\rho_{f_\kappa}\left( \frac{\pi y}{V^{\frac{2}{d}}} \right)dy.
\end{align}
By \eqref{rewrite3} and the definition of $g_V$, if $V$ is such that $g_V(y)\geq 0$ for a.e. $y\geq 1$ then
\begin{align}\label{ineq1}
&E_f^\kappa[V^{\frac{1}{d}}L]-E_f^\kappa[V^{\frac{1}{d}}L_d]+E_f^\kappa[V^{\frac{1}{d}}L^*]-E_f^\kappa[V^{\frac{1}{d}}L_d]\nonumber\\
&=\frac{\pi}{V^{\frac{2}{d}}}\int_1^\infty\left[ \theta_{L^*}\left( y\right)-\theta_{L_d}(y)\right]g_V(y)dy+\frac{\pi}{V^{\frac{2}{d}}}\int_1^\infty \left[ \theta_L(y)-\theta_{L_d}(y) \right]g_V(y)dy\nonumber\\
&\geq \frac{\pi}{V^{\frac{2}{d}}}\int_1^\infty m_L(y)g_V(y)dy,
\end{align}
where 
$$
m_L(y):=\min\{\theta_{L^*}\left( y\right)-\theta_{L_d}(y),  \theta_L(y)-\theta_{L_d}(y) \}.
$$
Since $m_L(y) \geq 0$ for all $L\in \mathcal{L}_d(1), y>0$ with equality if and only if $L=L_d$, and $g_V(y)\ge 0$ for a.e. $y\in [1,\infty)$, we get from \eqref{ineq1} that
\[
E_f^\kappa[V^{\frac{1}{d}}L]+E_f^\kappa[V^{\frac{1}{d}}L^*]\geq 2 E_f^\kappa[V^{\frac{1}{d}}L_d],\quad\mbox{with equality if and only if }L=L_d.
\]
It follows that $L_d$ is the unique minimizer of $L\mapsto E_f^\kappa[V^{\frac{1}{d}}L]$ on $\mathcal{L}_d(1)$, or equivalently that $V^{\frac{1}{d}}L_d$ is the unique minimizer of $E_f^\kappa$ in $\mathcal{L}_d(V)$, and the result is proved.
\end{proof}

The previous proof contains the main ingredients for showing Theorem \ref{thm02max}.

\begin{proof}[\textbf{Proof of Theorem \ref{thm02max}}]
Following exactly the same sequence of arguments as in the proof of the fourth point of Theorem \ref{thm1}, we obtain the maximality result of $V^{\frac{1}{d}}L_d^\pm$ at fixed density for $E_f^\pm$. Indeed, \eqref{jacobi} is replaced by
$$
\theta_L^\pm(\alpha)=y^{\frac{d}{2}}\theta_{L^*+c_{L^*}}(\alpha),
$$
and, by using the maximality of $L_d^\pm$ for $L\mapsto \theta_L^\pm(\alpha)$ and $L\mapsto \theta_{L+c_L}(\alpha)$ for all $\alpha>0$, we obtain
\begin{align}
&E_f^\pm [V^{\frac{1}{d}}L]-E_f^\pm[V^{\frac{1}{d}}L_d]+E_f^\pm[V^{\frac{1}{d}}L^*]-E_f^\pm[V^{\frac{1}{d}}L_d]\nonumber\\
&=\frac{\pi}{V^{\frac{2}{d}}}\int_1^\infty\left[ \theta_{L^*+c_{L^*}}\left( y\right)-\theta_{L_d^\pm +c_{L_d^\pm}}(y)\right]g_V(y)dy+\frac{\pi}{V^{\frac{2}{d}}}\int_1^\infty \left[ \theta_L^\pm (y)-\theta_{L_d^\pm}^\pm(y) \right]g_V(y)dy\nonumber\\
&\leq \frac{\pi}{V^{\frac{2}{d}}}\int_1^\infty m_L^\pm(y)g_V(y)dy,
\end{align}
where 
$$
m_L^\pm(y):=\max\{  \theta_{L^*+c_{L^*}}\left( y\right)-\theta_{L_d^\pm +c_{L_d^\pm}}(y),  \theta_L^\pm (y)-\theta_{L_d^\pm}^\pm(y)  \}.
$$
We again remark that $m_L^\pm(y)\leq 0$ for all $L\in \mathcal{L}_d(1)$, $y>0$ with equality if and only if $L=L_d^\pm$. Therefore, the positivity of $g_V$ as well as the universal maximality of $L_d^\pm$ implies in the same way that $V^{\frac{1}{d}}L_d^\pm$ is the unique maximizer of $E_f^\pm$ in $\mathcal{L}_d(V)$.
\end{proof}

The proof of Corollary \ref{Cor1} is a straightforward consequence of Theorem \ref{thm1}.

\begin{proof}[\textbf{Proof of Corollary \ref{Cor1}}]
Let $A_K:=\{a_k\}_{k\in K}\subset \R_+$ be such that $\mathsf{L}(A_K ,2)\leq 1$. Since $\mu_f\geq 0$, it follows that $\rho_f$ is positive, and furthermore  $\rho_f$ is increasing by assumption. Therefore, we have, for all $t>0$,
$$
\sum_{k \in K} \frac{a_k}{k^2}\rho_f\left( \frac{t}{\alpha^2} \right)\leq\sum_{k \in K} \frac{a_k}{k^2} \rho_f(t)=\mathsf{L}( A_K, 2)\rho_f(t)\leq \rho_f(t),
$$
where the first inequality is obtained from the monotonicity of $\rho_f$ and the last one from its positivity and the fact that $\mathsf{L}(A_K, 2)\leq 1$. The proof is completed by applying Theorem \ref{thm1}.
\end{proof}

We now show Theorem \ref{thm2ip} which is a simple consequence of the homogeneity of the Epstein zeta function and a property of the Riemann zeta function.

\begin{proof}[\textbf{Proof of Theorem \ref{thm2ip}}]
Using the homogeneity of the Epstein zeta function, we obtain
$$
E_f^{\kappa}[L]=\sum_{p\in L\backslash \{0\}} \frac{1}{|p|^{2s}}-\sum_{k \in K}\sum_{p\in L\backslash \{0\}}\frac{a_k}{k^{2s}|p|^{2s}}=\left(1-\mathsf{L}(A_K, 2s)\right)\zeta_L(2s),
$$
the exchange of sums being ensured by their absolute summability. If $\mathsf{L}( A_K, 2s)<1$, then $L\mapsto \zeta_L(2s)$ and $E_f^{\kappa}$ have exactly the same minimizer. If $\mathsf{L}(A_K, 2s)>1$, then the optimality are reversed and the proof is complete.

Furthermore, if $a_k=1$ for all $k \in K$, then we have
$$
\mathsf{L}(A_K, 2s)=\sum_{k \in K} \frac{1}{k^{2s}}\leq \zeta(2s)-1,
$$
where $\zeta(s):=\sum_{n\in \N} n^{-s}$ is the Riemann zeta function. Since $\zeta(x)< 2$ on $(0,\infty)$ if and only if $x> x_0\approx 1.73$, it follows that $\zeta(2s)-1<1$ if and only if $s>x_0/2\approx 0.865$ which is true for all $s>d/2$ whenever $d\geq 2$. We thus have $\mathsf{L}(A_K, 2s)<1$ and the proof is completed by application of point 1. of the theorem.
\end{proof}

Before proving Theorem \ref{thm3LJ}, we derive the following result, a generalization of our two-dimensional theorem \cite[Prop. 6.11]{BetTheta15}. Its proof follows the same main arguments as the two-dimensional version and it is a consequence of point 4. of Theorem \ref{thm1}.

\begin{prop}[\textbf{Optimality at high density for Lennard-Jones type potentials}]\label{prop:BP}
Let $f(r)=\frac{b_2}{r^{x_2}}-\frac{b_1}{r^{x_1}}$ where $b_1,b_2\in (0,\infty)$ and $x_2>x_1>d/2$, and let $L_d$ be universally optimal in $\mathcal{L}_d(1)$. If 
$$
V\leq \pi^{\frac{d}{2}} \left( \frac{b_2\Gamma(x_1)}{b_1\Gamma(x_2)} \right)^{\frac{d}{2(x_2-x_1)}},
$$ 
then $V^{\frac{1}{d}}L_d$ is the unique minimizer of $E_f$ in $\mathcal{L}_d(V)$.
\end{prop}
\begin{proof}[\textbf{Proof of Proposition \ref{prop:BP}}]
We follow the lines of \cite[Prop. 6.10]{BetTheta15} and we apply point 4. of Theorem \ref{thm1}. For $i\in \{1,2\}$, let $\beta_i:=b_i\frac{\pi^{x_i-1}}{\Gamma(x_i)}$ and $\alpha:=V^{\frac{2}{d}}$, then $g_V(y)=\frac{y^{\frac{d}{2}-x_2-1}}{\alpha^{x_1-1}}\tilde{g}_V(y)$ where $g_V$ is given by \eqref{eq:condgV} and
$$
\tilde{g}_V(y):=\frac{\beta_2}{\alpha^{x_2-x_1}}y^{2x_2-\frac{d}{2}}-\beta_1 y^{x_2+x_1-\frac{d}{2}}-\beta_1 y^{x_2-x_1} + \frac{\beta_2}{\alpha^{x_2-x_1}}.
$$
We therefore compute $\tilde g_V'(y)=y^{x_2-x_1-1}u_V(y)$ where 
$$
u_V(y):=\beta_2\left(2x_2-\frac{d}{2}\right)\frac{y^{x_2+x_1-\frac{d}{2}}}{\alpha^{x_2-x_1}}-\beta_1\left(x_2+x_1-\frac{d}{2}\right)y^{2x_1-\frac{d}{2}}-\beta_1(x_2-x_1).
$$
Differentiating again, we obtain 
$$
u_V'(y)=\left( x_2+x_1-\frac{d}{2} \right)y^{2x_1-\frac{d}{2}-1}\left(\beta_2\left(2x_2-\frac{d}{2}\right)\frac{y^{x_2-x_1}}{\alpha^{x_2-x_1}}-\beta_1\left(2x_1-\frac{d}{2}\right)  \right),
$$
and we have that $u_V'(y)\geq 0$ if and only if $y\geq \left( \frac{\beta_1(2x_1-\frac{d}{2})}{\beta_2(2x_2-\frac{d}{2})} \right)^{\frac{1}{x_2-x_1}}\alpha$. By assumption, we know that
$$
\alpha\leq \pi\left(\frac{a_2\Gamma(x_1)}{a_1\Gamma(x_2)}  \right)^{\frac{1}{x_2-x_1}}=\left( \frac{\beta_2}{\beta_1}\right)^{\frac{1}{x_2-x_1}}<\left( \frac{\beta_2(2x_2-\frac{d}{2})}{\beta_1(2x_1-\frac{d}{2})}\right)^{\frac{1}{x_2-x_1}},
$$
which implies that $u_V'(y)\geq 0$ for all $y\geq 1$. We now remark that
$$
u_V(1)=\left(2x_2-\frac{d}{2}\right)\left(\frac{\beta_2}{\alpha^{x_2-x_1}}-\beta_1  \right)\geq 0,
$$
by assumption, since $p>d/2>d/4$ and  
\begin{equation}\label{eq:equivalphagamma}
\displaystyle \alpha\leq \pi\left(\frac{b_2\Gamma(x_1)}{b_1\Gamma(x_2)}  \right)^{\frac{1}{x_2-x_1}}\iff \frac{\beta_2}{\alpha^{x_2-x_1}}-\beta_1  \geq 0.
\end{equation}
It follows that $g_V'(y)\geq 0$ for all $y\geq 1$. Since
$$
g_V(1)=2\left( \frac{\beta_2}{\alpha^{x_2-1}} - \frac{\beta_1}{\alpha^{x_1-1}}\right)\geq 0
$$
again by \eqref{eq:equivalphagamma}, $g_V(y)\geq 0$ for all $y\geq 1$ and the proof is complete.
\end{proof}

\begin{proof}[\textbf{Proof of Theorem \ref{thm3LJ}}]
Let $A_K=\{a_k\}_{k\in K}$ for some $K\subset \N\backslash \{1\}$ and $f(r)=c_2 r^{-x_2}-c_1r^{-x_1}$, then we have, using the homogeneity of the Epstein zeta function,
\begin{align*}
E_f^{\kappa}[L]&=c_2 \zeta_L(2x_2)-c_1\zeta_L(2x_1)-\sum_{k\in K} a_k\left( c_2\zeta_{kL}(2x_2)-c_1\zeta_{kL}(2x_1) \right)\\
&=c_2\left(1-\mathsf{L}( A_K, 2x_2)\right)\zeta_L(2x_2)-c_1\left(1-\mathsf{L}(A_K, 2x_1))\right)\zeta_L(2x_1).
\end{align*}
We now assume that $\mathsf{L}(A_K, 2x_2)<\mathsf{L}(A_K, 2x_1)<1$. Therefore, the first part of point 1. is a simple consequence of Prop.\ref{prop:BP} applied for the coefficients $b_i=c_i\left( 1-\mathsf{L}( A_K, 2x_i) \right)>0$ where $i\in \{1,2\}$. The fact that $E_f^{\kappa}$ is not minimized by $L_d$ for $V$ large enough is a direct application of \cite[Thm. 1.5(1)]{OptinonCM} since $\mu_f$ is negative on $(0,r_0)$ for some $r_0$ depending on the parameters $c_1,c_2,x_1,x_2,A_K$. Furthermore, the fact that the shape of the minimizers are the same follows from \cite[Thm. 1.11]{OptinonCM} where it is shown that the minimizer of the Lennard-Jones type lattice energies does not depend on the coefficients $b_1,b_2$ but only on the exponents $x_1,x_2$, which are the same for $f$ and $f_{\kappa}$.

If $\mathsf{L}( A_K, 2x_1)>\mathsf{L}( A_K, 2x_2)>1$, then $f_{\kappa}(r)=-b_2 r^{-x_2}+b_1 r^{-x_1}$ where $b_i:=c_i\left( \mathsf{L}( A_K, 2x_i)-1\right)>0$, $i\in \{1,2\}$. If follows that $f_{\kappa}(r)$ tends to $-\infty$ as $r\to 0$, which implies the same for $E_f^{\kappa}[L]$ as $L$ has its lengths going to $0$ and $+\infty$, i.e. when $L$ collapses. This means that $E_f^{\kappa}$ does not have a minimizer in $\mathcal{L}_d(V)$ and in $\mathcal{L}_d$. Furthermore, combining point 1. with the fact that the signs of the coefficients are switched, we obtain the maximality of $V^{1/d} L_d$ at high density (i.e. low volume $V<V_\kappa$).

If $\mathsf{L}( A_K, 2x_1)>1>\mathsf{L}(A_K, 2x_2)$, then $f_{\kappa}(r)=b_2 r^{-x_2}+b_1 r^{-x_1}$ where $b_1:=c_1\left( \mathsf{L}( A_K, 2x_1)-1\right)>0$ and $b_2:=c_2\left(1-\mathsf{L}( A_K, 2x_2)\right)>0$. Therefore $f_{\kappa}\in \mathcal{F}_d^{cm}$, which implies the optimality of $V^{1/d} L_d$ in $\mathcal{L}_d(V)$ for all fixed $V>0$ and the fact that $E_f^{\kappa}[L]$ tends to $0$ as all the points are sent to infinity, i.e. $E_f^{\kappa}$ does not have a minimizer in $\mathcal{L}_d$.

\end{proof}

\noindent\textbf{Acknowledgement:} I am grateful for the support of the WWTF research project "Variational Modeling of Carbon Nanostructures" (no. MA14-009) and the (partial) financial support from the Austrian Science Fund (FWF) project F65. I also thank Mircea Petrache for our discussions about the crystallization at fixed density as a consequence of Cohn-Kumar conjecture stated in Remark \ref{rmk:CK}.

\newpage

{\small \bibliographystyle{plain}
\bibliography{Sparse}}

\end{document}